\documentclass[aps,prl,showpacs,twocolumn,nofootinbib,10pt]{revtex4-1}
 \pdfoutput=1

\usepackage{graphicx}
\usepackage{amsmath}
\usepackage{amssymb}
\usepackage{mathrsfs}
\usepackage{amsthm}
\usepackage{bm}
\usepackage{url}
\usepackage[T1]{fontenc}
\usepackage{csquotes}
\MakeOuterQuote{"}
\usepackage{chapterbib}
\newtheoremstyle{note}
  {\topsep/2}               % ABOVE SPACE
  {\topsep/2}               % BELOW SPACE
  {}                      % BODY FONT
  {\parindent}            % INDENT (empty value is the same as 0pt)
  {\itshape}              % HEAD FONT
  {.}                     % HEAD PUNCTUATION
  {5pt plus 1pt minus 1pt}% HEAD SPACE
  {}

\theoremstyle{note}
\newtheorem{theorem}{Theorem}
\newtheorem{lemma}{Lemma}

\theoremstyle{definition}

\theoremstyle{remark}
\newtheorem{remark}{Remark}

\newcommand{\vecrm}[1]{\mathbf{#1}}
\newcommand{\mrm}[1]{\mathrm{#1}}
\newcommand{\tr}{\operatorname{tr}}

\newcommand{\rmi}{\mathrm{i}}
\newcommand{\rme}{\mathrm{e}}

\newcommand{\rmd}{\mathrm{d}}

\newcommand{\be}{\begin{equation}}
\newcommand{\ee}{\end{equation}}
\newcommand{\ba}{\begin{align}}
\newcommand{\ea}{\end{align}}

\def\<{\langle}  %% overiding the original command \<
\def\>{\rangle}  %% overiding the original command \>

\newcommand{\bbF}{\mathbb{F}}
\newcommand{\Sp}[2]{\mrm{Sp}(#1,#2)}

\newcommand{\SL}[2]{\mrm{SL}(#1,#2)}

\newcommand{\GL}[2]{\mrm{GL}(#1,#2)}

\newcommand{\PSU}[2]{\mrm{PSU}(#1,#2)}

\newcommand{\GaL}[2]{\mrm{\Gamma L}(#1,#2)}

\newcommand{\vp}{V_\mrm{0}}

\newcommand{\phw}{\overline{D}}

\newcommand{\pc}{\overline{\mathrm{C}}}

\def\eqref#1{\textup{(\ref{#1})}}  %% overiding the original command \eqref

\newcommand{\thref}[1]{Theorem~\ref{#1}}
\newcommand{\Thref}[1]{Theorem~\ref{#1}}

\newcommand{\Thsref}[1]{Theorems~\ref{#1}}

\newcommand{\lref}[1]{Lemma~\ref{#1}}

\newcommand{\cref}[1]{Conjecture~\ref{#1}}
\newcommand{\Cref}[1]{Conjecture~\ref{#1}}

\newcommand{\rcite}[1]{Ref.~\cite{#1}}
\newcommand{\rscite}[1]{Refs.~\cite{#1}}

\begin{document}

\title{Permutation Symmetry Determines  the Discrete Wigner Function}
\author{Huangjun Zhu}
\affiliation{Perimeter Institute for Theoretical Physics, Waterloo, On N2L 2Y5,
Canada}
\email{hzhu@pitp.ca}
\affiliation{Institute for Theoretical Physics, University of Cologne, 
Cologne 50937, Germany}
\email{hzhu1@uni-koeln.de}

\pacs{03.67.-a, 03.65.-w, 02.10.De}

%03.67.-a quantum information
%03.65.Wj quantum tomography, state reconstruction
%02.10.De algebraic structure
%03.65.-w quantum mechanics

%03.67.Mn Entanglement production, characterization and manipulation
%03.65.Ud Entanglement and quantum nonlocality
%(e.g. EPR paradox, Bell's inequalities, GHZ states, etc.)
%(for entanglement production in quantum information, see 03.67.Mn);

\begin{abstract}
The Wigner function provides a useful quasiprobability representation of quantum mechanics, with applications in various branches of physics. Many nice properties of the Wigner function are intimately connected with the high symmetry of the underlying operator basis composed of phase point operators: any pair of phase point operators can be transformed to any other pair by a unitary symmetry transformation. We prove that, in the discrete scenario, this permutation symmetry is equivalent to the symmetry group being a unitary 2-design.
Such a highly  symmetric representation can only appear in odd prime power dimensions besides  dimensions 2 and 8.
It  suffices to single out a unique discrete Wigner function among all possible quasiprobability representations.
In the course of our study, we  show
that this  discrete Wigner function is   uniquely determined by Clifford covariance, while no  Wigner function  is Clifford covariant in any even prime power dimension. 
\end{abstract}

\date{\today}
\maketitle

\begin{cbunit}

The Wigner quasiprobability distribution in phase space, originally introduced for studying quantum correction to thermodynamics, has numerous applications in various branches of physics, such as quantum optics, quantum chaos, and quantum computing. It also  provides an alternative formulation of quantum mechanics, which is particularly suitable  for studying quantum-classical correspondence \cite{HillOSW84, KimZ87, CurtFZ14}. The usefulness of the Wigner function is intimately connected to the high symmetry of the underlying operator basis composed of \emph{phase point operators}. Reminiscent of the symplectic geometry in classical phase space, this basis is invariant under displacements and linear canonical transformations, which realize translations, rotations, and squeezing in phase space.  Therefore, "no point, no direction, no scale is distinguished from any other in phase space" \cite{Engl89}.  This  means that the symmetry group of the basis acts doubly transitively on  phase point operators; that is, any pair of distinct phase point operators  can be turned  into any other pair by a unitary symmetry transformation.

Recently, many  discrete analogues of the Wigner function have been introduced and found various applications in quantum information science, such as quantum state tomography and quantum computation \cite{Woot87,Leon95, GibbHW04,VeitFGE12, HowaWVE14}. In addition,  general quasiprobability representations have been found useful for studying quantum foundations \cite{Hard01,FuchS13}.
At this point, it is natural to reflect on the following questions: what is so special about the Wigner function? Is there a simple criteria that can single out a particular quasiprobability representation among all potential candidates?  Although similar questions have been investigated extensively \cite{HillOSW84, Woot87, GibbHW04, Gros06}, no satisfactory answer is known, especially in the discrete scenario. In every odd prime dimension, the Wootters discrete Wigner function  \cite{Woot87} is uniquely determined by Clifford covariance \cite{Gros06}.  However, the situation is not clear for other dimensions, despite the intensive efforts of many researchers over the last three decades.

In this paper we show that the operator basis underlying the Wootters discrete Wigner function is almost uniquely characterized by the permutation symmetry pertaining to the continuous analogue; that is, any pair of distinct phase point operators can be transformed to any other pair by a unitary transformation.
This criterion is motivated by  symplectic geometry of the classical phase space and is based on intrinsic symmetry of phase point operators, 
so it applies equally well to  Wigner functions and general quasiprobability representations, thereby facilitating their comparison.   Our study demonstrates that the Wootters discrete Wigner function is the most symmetric quasiprobability representation of quantum mechanics. Our study also reveals the  group theoretical root why such a highly symmetric  representation can only exist in odd prime power dimensions besides   dimensions 2 and 8, thereby resolving the enigma on the discrete Wigner function that persists  for the past three decades. The  exception for dimension~8 is tied with  a special symmetric informationally
complete measurement (SIC) \cite{Zaun11, ReneBSC04, ScotG10,  ApplFZ15G},  known as Hoggar lines \cite{Hogg98, Zhu12the, Zhu15S}, which is of independent interest.

In the course of our study, we show that an operator basis has doubly transitive permutation symmetry  if and only if its symmetry group is a unitary 2-design \cite{Dank05the, DankCEL09, GrosAE07,Zhu15MC,Webb15}.
Therefore, the basis of phase point operators is almost uniquely characterized by its symmetry group being a unitary 2-design. In addition, in every odd prime power dimension, it is  the only operator basis up to scaling that is Clifford covariant, while no such operator basis exists in any even prime power dimension.  The first conclusion establishes a surprising connection between the physics of phase space and a ubiquitous tool in quantum information science.
The latter conclusion generalizes a result of Gross~\cite{Gros06} and settles the long-standing open problem on Clifford covariant discrete Wigner functions.  These results  may have profound implications for quantum information,  quantum computation, and  quantum foundations, given the prominent roles played by discrete Wigner functions, unitary 2-designs, and the Clifford group. In particular, the existence of a Clifford covariant discrete Wigner function is crucial to understanding many  fascinating subjects, including computational speedup and contextuality   \cite{VeitFGE12, HowaWVE14}.

To set up the stage, we need to introduce the (multipartite) Heisenberg-Weyl (HW) group.  In prime dimension $p$,  the HW group $D$ is
generated by the phase operator $Z$
 and the cyclic-shift operator $X$,
\begin{equation}\label{eq:HW}
Z|u\rangle=\omega^u |u\rangle, \quad X|u\rangle=|u+1\rangle,
\end{equation}
where $\omega=\mathrm{e}^{2\pi \mathrm{i}/p}$,
$u\in \bbF_p$, and $\bbF_p$  is the field
of integers modulo $p$.
In prime power dimension $q=p^n$, the  HW group $D$ is the tensor power of $n$ copies of the HW group in  dimension $p$. The elements in the HW are called displacement operators. Up to phase factors, they can be labeled by  vectors in a symplectic space of dimension $2n$ over $\mathbb{F}_p$ as $
D_{\mu}= \tau^{\sum_j \mu_j \mu_{n+j} }\prod_{j=1}^n X_j^{\mu_j}  Z_j^{\mu_{n+j}}$,
where  $\tau=-\rme^{\pi \rmi/p}$,  while $Z_j$ and $X_j$ are the phase operator and cyclic shift operator of the $j$th party.

The \emph{Clifford group} $\mathrm{C}$ is  composed of  unitary operators that map displacement operators to displacement operators up  to phase factors. Any Clifford unitary $U$ induces a symplectic transformation on the symplectic space that labels displacement operators \cite{BoltRW61I,BoltRW61II,Zhu15Sh}. The quotient $\pc/\phw$ ($\overline{G}$ denotes the group $G$ modulo phase factors) can be identified with the symplectic group $\Sp{2n}{p}$. 
The Clifford group plays a fundamental role in various branches of quantum
information science, such as quantum computation,  quantum error correction,  quantum tomography, and randomized
benchmarking. Many of these applications  are tied with the fact that the Clifford group is a unitary 2-design \cite{Dank05the, DankCEL09, GrosAE07}. When $p=2$, it  is also a unitary 3-design according to a recent result of the author \cite{Zhu15MC} and that of Webb \cite{Webb15}. Recall
that a set of $K$ unitary operators
$\{U_j\}$ is a  \emph{unitary $t$-design} \cite{Dank05the, DankCEL09, GrosAE07,Zhu15MC,Webb15}
if it satisfies
\begin{equation}\label{eq:U2design}
\frac{1}{K} \sum _j U_j^{\otimes t} M(U_j^{\otimes t})^\dag =\int \rmd U
U^{\otimes t}M(U^{\otimes t})^\dag
\end{equation}
for any operator $M$ acting on the $t$-partite Hilbert space, where the integral
is taken over the whole unitary group with respect to the normalized Haar
measure. The Clifford group and unitary 2-designs are interesting here because of their close connections with the Wootters discrete Wigner function, as we shall see shortly.

 The center of $\Sp{2n}{p}$ for odd prime $p$ is generated by the scalar matrix $-1$, referred to as the \emph{central involution} henceforth (an involution is simply an element of order~2). The $q^2$  Clifford unitaries that induce the central involution all have order 2 and  are called \emph{principal involutions}.
With  suitable  phase factors,  the principal involutions happen to be the phase point operators underlying the Wootters discrete Wigner function \cite{Woot87,Gros06}.
A special phase point operator is the \emph{parity operator} $\vp$,
\begin{equation}
\vp|\vecrm{u}\rangle=|-\vecrm{u}\rangle, \quad \vp D_\mu \vp^\dag =D_{-\mu},
\end{equation}
where the vectors $\vecrm{u}\in \bbF_p^n$ label the elements in the computational basis.  Other phase point operators are basically
 displaced parity operators,
\begin{equation}
V_\mu=D_\mu \vp D_\mu^\dag=D_{2\mu}V_0.
\end{equation}
Each $V_\mu$ has $(q+1)/2$ eigenvalues equal to 1 and $(q-1)/2$ eigenvalues equal to $-1$, so $\tr(V_\mu)=1$ and $\tr(V_\mu^2)=q$.

The set of  phase point operators $V_{\mu}$ forms an orthogonal basis (and  a unitary error basis) in the operator space.  This basis is invariant under the whole Clifford group, so any pair of distinct phase point operators can be transformed to any other pair by a unitary transformation as pointed out earlier.
Any state $\rho$ can be expanded in  phase point operators
$\rho=\sum_\mu W_\mu V_\mu$, and the coefficients $W_{\mu} =\tr(\rho V_{\mu})/q$ define the Wootters discrete Wigner function, which is Clifford covariant \cite{Woot87, Gros06}.

An \emph{operator frame} $\{F_j\}$ on a Hilbert space is a set of  operators that span the operator space. Frames are interesting to us because they provide a unified framework for studying general quasiprobability representations of quantum mechanics \cite{FerrE08,FerrE09}. Any operator frame in dimension $d$ has at least $d^2$ elements; those attaining the lower bound are   minimal. The symmetry group of an operator frame is composed of all unitary operators $U$ that leave the frame invariant, that is, $UF_jU^\dag=F_{\sigma(j)}$, where $\sigma$ is a permutation; unitary operators that differ only by overall phase factors are identified, but the phases of frame elements do matter. This group is of paramount importance because the symmetry of the frame is responsible for the symmetry and usefulness of the derived quasiprobability representation, as manifested in the Wigner function. The frame $\{F_j\}$
is   \emph{group covariant} if its symmetry group acts transitively on the frame elements. It is  \emph{supersymmetric} if the symmetry group acts doubly transitively, that is, any ordered pair of distinct frame elements can be mapped to any other pair. Similar terminology applies to operator bases and positive-operator-valued measures (POVMs).

A frame  is called a Wigner-Wootters frame (or basis) if it has the form  $\{a+bV_\mu\}$ with $V_\mu$ phase point operators and $a, b$   constants. According to the above discussion, any  Wigner-Wootters frame is supersymmetric. Remarkably, the converse is also true except for frames constructed from two special SICs. Recall that a SIC in dimension $d$
is composed of $d^2$ subnormalized projectors onto pure states
$\Pi_j=|\psi_j\rangle\langle\psi_j|$ with equal pairwise fidelity $|\langle\psi_j|\psi_k\rangle|^2=(d\delta_{jk}+1)/(d+1)$
\cite{Zaun11,ReneBSC04, ScotG10,Zhu12the, ApplFZ15G}.
 A frame of the form $\{a+b\Pi_j\}$ is called a SIC frame (or basis). A SIC frame is supersymmetric if and only if the corresponding SIC is supersymmetric. According to a recent result of the author \cite{Zhu15S}, the SIC in dimension~2, the Hesse SIC in dimension~3, and  the set of Hoggar lines in dimension~8 are the only three supersymmetric SICs (super-SICs in short). The  frames corresponding to the three SICs are referred to as tetrahedron frames, Hesse frames, and Hoggar frames, respectively. Similar terminology applies to operator bases and POVMs so constructed.
 Interestingly, a Hesse frame is also a Wigner-Wootters frame; note that the projectors onto one-dimensional eigenspaces of phase point operators in dimension 3 form the Hesse SIC \cite{Zhu15S}. 

\begin{theorem}\label{thm:2designSSY}
An operator basis is supersymmetric if and only if its symmetry group is
a unitary 2-design.
\end{theorem}

\begin{theorem}\label{thm:WignerClifford}
In any odd prime power dimension, an operator basis is Clifford covariant  if and only if it is a Wigner-Wootters basis. The Wootters
discrete Wigner function is the unique Clifford covariant discrete Wigner function.
No such operator basis or discrete Wigner function
exists in any even prime power dimension.
\end{theorem}

\begin{theorem}[CFSG]\label{thm:SuperFrame}
Any supersymmetric operator frame is unitarily equivalent to a Wigner-Wootters frame except for tetrahedron frames in dimension 2 and Hoggar frames in dimension 8.
\end{theorem}

 \Thsref{thm:2designSSY}, \ref{thm:WignerClifford}, and  \ref{thm:SuperFrame}
establish spectacular connections among many interesting subjects, including discrete Wigner functions,   unitary 2-designs,  supersymmetric operator bases, and the Clifford group.  They also settle a number of persistent open problems concerning discrete Wigner functions and the Clifford group. In particular, \thref{thm:WignerClifford} generalizes a result of Gross  \cite{Gros06} and settles the  open problem on Clifford covariant discrete Wigner functions.
\Thsref{thm:2designSSY} and \ref{thm:SuperFrame} further establish two appealing characterizations of the basis of phase point operators: 
this basis  is almost uniquely determined by the permutation symmetry or the characteristic that its symmetry group is a unitary 2-design.

 The proofs of \Thsref{thm:2designSSY} and \ref{thm:WignerClifford} turn out to be surprisingly simple, as we shall see shortly. \Thref{thm:SuperFrame} does not assume that the   frame is minimal; rather
this condition  is a consequence of the permutation symmetry, which can be
verified by inspecting the Gram matrix. The proof of \thref{thm:SuperFrame} relies on the classification of 2-transitive permutation groups, which in turn relies on the classification of finite simple groups (CFSG) \cite{DixoM96book, Came99book, Zhu15S}. To manifest this point, all theorems and lemmas are marked with "CFSG" whenever CFSG is involved in the proofs. However,  knowledge of CFSG is not necessary to understand our reasoning, whose  basic idea  is  pretty simple.

To prove \Thsref{thm:2designSSY} and \ref{thm:WignerClifford}, we need to introduce two simple but  useful lemmas.  
\begin{lemma}\label{lem:orbits}
Suppose $\overline{G}$ is a subgroup of the symmetry group of an operator basis. Then
the number of orbits  of
$\overline{G}$ on the basis elements is equal to the sum of squared multiplicities of
all the
inequivalent irreducible components of $\overline{G}$. In particular, $\overline{G}$ acts transitively
on the basis if and only if it is irreducible.
\end{lemma}
\begin{remark}
This lemma  follows from the same  reasoning as  the proof of Lemma 7.2 in the author's thesis \cite{Zhu12the} (see also  \rscite{Zhu15Sh,Zaun11}).
Here we present a simpler proof.
\end{remark}
\begin{proof}
The  number of orbits  of
$\overline{G}$ on the basis elements is equal to the dimension of the commutant (in the operator space)
of $G$. According to representation theory, this dimension is equal to the
sum of squared multiplicities of all
inequivalent irreducible components of $\overline{G}$.
\end{proof}

\begin{lemma}\label{lem:2transitiveStab}
A group covariant operator basis is  supersymmetric if and only if the  stabilizer (within the symmetry group)  of each basis element 
has two  irreducible components,  which are inequivalent.
\end{lemma}
\begin{proof}
The stabilizer of each basis element of a supersymmetric operator  basis acts transitively on the remaining basis
elements, so it has two orbits  and thus  two inequivalent irreducible components. Conversely, if the stabilizer has  two inequivalent irreducible components, then it has two orbits on the basis elements, so the basis is supersymmetric. 
\end{proof}

\begin{proof}[Proof of \thref{thm:2designSSY}]
Let $\{F_j\}$ be an operator basis for the  space $\mathcal{H}$ with symmetry group
$\overline{G}$. Then $\{F_j\otimes F_k\}$ is an operator basis for $\mathcal{H}\otimes
\mathcal{H}$.
The basis $\{F_j\}$ is supersymmetric if and only if the group $\overline{G}_{(2)}:=\{U\otimes U|U\in \overline{G}\}$
has two orbits on $\{F_j\otimes F_k\}$.  According to \lref{lem:orbits},
the latter is equivalent to the condition that $\overline{G}_{(2)}$ has two inequivalent
irreducible components, which holds if and only if   $\overline{G}$ is a unitary 2-design \cite{GrosAE07, Zhu15MC}.
\end{proof}

\begin{proof}[Proof of \thref{thm:WignerClifford}]

Let $\{F_j\}$ be an operator basis that is covariant with respect to the Clifford group; then it is   also covariant with respect to the HW group. The stabilizer $\overline{S}$ of each basis element, say $F_1$, forms a complement of the HW group within the Clifford group and is isomorphic to the symplectic group $\Sp{2n}{p}$. In addition, $\{F_j\}$ is necessarily supersymmetric, so  $\overline{S}$ has two inequivalent irreducible components according to \lref{lem:2transitiveStab}.

When   $p$ is odd, $\overline{S}$ must contain a  principal
involution $\overline{U}$ in its center. With a suitable choice of the phase
factor, $U$ is a phase point operator. Since $U$ has nonzero trace,  it follows
that all elements of $S$ commute with $U$. Therefore,
the two irreducible components of $\overline{S}$ correspond to the two eigenspaces
of $U$. Accordingly, $F_1$ is a linear combination of the two projectors onto
the two eigenspaces or, equivalently, a linear combination of $U$ and the
identity; so $\{F_j\}$ is a Wigner-Wootters basis.

When $p=2$ and $n=1$, it is straightforward to verify that any supersymmetric  operator basis is a tetrahedron basis, which cannot be covariant with respect to the Clifford group. When $p=2$ and  $n\geq2$, the HW group
 is not complemented in the  Clifford group according to theorem~7 in \rcite{BoltRW61II}, so  the Clifford group  has no subgroup isomorphic to $\Sp{2n}{2}$.
When $n\geq4$, this conclusion also follows from the fact that the
minimal degree of  nontrivial irreducible projective representations of $\Sp{2n}{2}$
is $(2^n-1)(2^{n-1}-1)/3$ according to Table~II in \rcite{TiepZ96},  which is larger than $2^n$. Consequently, no operator basis can be covariant with respect to the Clifford
group.
\end{proof}

To prove  \thref{thm:SuperFrame}, we need to introduce additional tools.  The following theorem establishes a deep connection between permutation symmetry and the HW group, which  is of independent interest.\begin{theorem}[CFSG]\label{thm:SuperFrameHW}
Every supersymmetric operator frame is covariant with  a multipartite HW group; its  symmetry group is a subgroup of  the  Clifford group.
\end{theorem}
\begin{proof}
The proof is similar to  the proof of Theorem~4 (concerning super-SICs) in \rcite{Zhu15S}.
Suppose $\{F_j\}$ is a supersymmetric operator frame  in dimension $d$ with symmetry group $\overline{G}$. Then $\{F_j\}$  has $d^2$ elements, and $\overline{G}$ acts doubly transitively on the frame elements.  According to Burnside's theorem on 2-transitive permutation groups \cite{DixoM96book, Came99book, Zhu15S}, $\overline{G}$ has a unique  minimal normal subgroup, say $\overline{N}$, which
is either regular   elementary Abelian or primitive non-Abelian simple. In either case, $\overline{N}$  is irreducible according to \lref{lem:orbits}. The latter case can be excluded by the same reasoning as in the proof of Lemma~14 in \rcite{Zhu15S} (involving the CFSG). In the former case, $\overline{N}$ is a faithful irreducible projective representation of an elementary Abelian group, so it is  (projectively) unitarily   equivalent to  the   HW group in a prime power dimension (see Lemma~1 in \rcite{Zhu15S}). Since  $\overline{N}$ is normal in $\overline{G}$, it follows  that  $\overline{G}$  is a subgroup of the   Clifford group.
\end{proof}

Let $\{F_j\}$ be a supersymmetric operator frame in dimension $q=p^n$ and $\overline{S}$  the stabilizer of one of the frame elements, say $F_1$. According to \thref{thm:SuperFrameHW},  $\overline{S}$ can be identified with a transitive subgroup of the Clifford group, where "transitive" means that $\overline{S}$ acts (by conjugation) transitively on nontrivial displacement operators. In addition, $\overline{S}$ has trivial intersection with the HW group, so it is isomorphic to a transitive subgroup of $\Sp{2n}{p}$ (a subgroup that acts transitively on nonzero vectors in $\bbF_p^{2n}$). To establish our main result, it is therefore crucial to figure out transitive subgroups of $\Sp{2n}{p}$. Transitive linear groups have already been  classified by Hering \cite{Heri85} and Liebeck \cite{Lieb87}; see also \rcite{DixoM96book} and Table 7.3 in \rcite{Came99book}.
Quite surprisingly, barring a few exceptions in low dimensions, these transitive subgroups divide into a few families, which have pretty simple structures. For odd $p$  in particular, all we need is the following simple fact; see Supplementary Material \cite{Zhu16PS}.

\begin{lemma}[CFSG]\label{lem:TransitiveOdd}
Any transitive  subgroup $H$ of $\Sp{2n}{p}$ for odd prime $p$ contains the central involution.
\end{lemma}
\begin{remark}
Besides its significance to the present study, this lemma is very useful in understanding the structure of unitary 2-designs, which are  intimately connected with transitive subgroups of the symplectic group \cite{Gros06, Zhu15MC,Zhu15U}.
\end{remark}

\begin{proof}[Proof of \thref{thm:SuperFrame}]
 Any supersymmetric frame $\{F_j\}$ in dimension $d$ is necessarily minimal and  forms an operator basis.  According to \thref{thm:SuperFrameHW}, $d$ is  a prime power $q=p^n$,   the frame $\{F_j\}$ is covariant with respect to the  HW group, and its symmetry group is a subgroup of the Clifford group that acts doubly transitively on  frame elements. Let $\overline{S}$ be the stabilizer (within the Clifford group) of the  frame element $F_1$, then $\overline{S}$  is  transitive and it induces a transitive subgroup $H$ of $\Sp{2n}{p}$. In addition, $\overline{S}$ has two inequivalent irreducible components by \lref{lem:2transitiveStab}.

When $p$ is odd, $H$ contains the central involution of $\Sp{2n}{p}$  according to \lref{lem:TransitiveOdd}, so that $\overline{S}$ contains a  principal involution $\overline{U}$ in its center. According to the same reasoning as in the proof of \thref{thm:WignerClifford}, $\{F_j\}$ is a Wigner-Wootters frame.

When  $p=2$, let $h$ be the smaller degree of the two irreducible components of $\overline{S}$. If $h=1$, then  the projector onto the irreducible component of degree 1 is a fiducial state for a super-SIC. According to \rcite{Zhu15S} (see also \rcite{GodsR09}), the only super-SICs in even prime power dimensions are the SIC in dimension 2 and the set of Hoggar lines in dimension 8. Therefore, $\{F_j\}$
is either a tetrahedron frame or a Hoggar frame. It remains to consider the scenario $2\leq h\leq 2^n/2$   with  $n\geq2$.
Calculation shows that the minimal degree of  nontrivial irreducible projective representations of $H$  is always larger than $2^n/2$ except when  $n=2$ and $H$ is isomorphic to $\Sp{2}{2^2}\simeq\SL{2}{2^2}$; see the supplementary material.
In the exceptional case, any subgroup of the Clifford group  that is  isomorphic to $\SL{2}{2^2}$ is either transitive and irreducible, or nontransitive and reducible.
\end{proof}

\Thsref{thm:WignerClifford} and \ref{thm:SuperFrame}  imply that the Hesse SIC is the unique Clifford covariant SIC; cf.~\rscite{Zhu10,Zhu15S}.  \Thref{thm:SuperFrame} generalizes  a recent  result of the author on super-SICs~\cite{Zhu15S} to  POVMs that are not necessarily rank 1.
\begin{theorem}[CFSG]
Any supersymmetric POVM is a Wigner-Wootters POVM except for tetrahedron POVMs in dimension 2 and Hoggar POVMs in dimension 8.
\end{theorem}

In summary, we  showed that  the Wootters discrete Wigner function is almost uniquely characterized by the permutation symmetry appearing in the continuous analogue. This permutation symmetry amounts to the requirement that the symmetry group of the underlying operator basis is a unitary 2-design.  In addition, this Wigner function  is  the only choice that is Clifford covariant in any odd prime power dimension, while no such Wigner function can exist in any even prime power dimension.
Our study settles several long-standing open problems on discrete Wigner functions, including the one concerning the existence and uniqueness of Clifford covariant discrete Wigner functions. Our study also establishes surprising connections among various fascinating subjects: such as discrete Wigner functions, unitary 2-designs, Clifford groups, and permutation symmetry. These results provide valuable insight  on   a number of fundamental issues in quantum information and quantum foundations, in particular the distinction between qubit and qudit stabilizer formalisms.
In addition, the technical tools developed in this paper are useful for studying  many discrete symmetric structures behind finite-state quantum mechanics, such as symmetric POVMs, mutually unbiased bases, and unitary $t$-designs. An interesting problem left open is whether similar results hold in the continuous scenario.

\section*{Acknowledgments}
The author is grateful to Nick Gill for introducing  Singer cycles and Zsigmondy primes, and to Robert Spekkens,  Mark Howard, and Joseph Emerson for discussions. This work is supported by Perimeter Institute for Theoretical Physics. Research at Perimeter Institute is supported by the Government of Canada through Industry Canada and by the Province of Ontario through the Ministry of Research and Innovation.  The author also acknowledges financial support
from the Excellence Initiative of the German Federal and State Governments
(ZUK 81) and the DFG.

\end{cbunit}

\clearpage

%%%%%%%%% Merge with supplemental materials %%%%%%%%%%
%%%%%%%%% Prefix a "S" to all equations, figures, tables and reset the counter
%%%%%%%%%
\setcounter{equation}{0}
\setcounter{figure}{0}
\setcounter{table}{0}
\setcounter{theorem}{0}
\setcounter{lemma}{0}
\setcounter{remark}{0}

\makeatletter
\renewcommand{\theequation}{S\arabic{equation}}
\renewcommand{\thefigure}{S\arabic{figure}}
\renewcommand{\thetable}{S\arabic{table}}
\renewcommand{\thetheorem}{S\arabic{theorem}}
\renewcommand{\thelemma}{S\arabic{lemma}}
\renewcommand{\theremark}{S\arabic{remark}}

%\renewcommand{\bibnumfmt}[1]{[S#1]}
%\renewcommand{\citenumfont}[1]{S#1}
%%%%%%%%%% Prefix a "S" to all equations, figures, tables and reset the counter
%%%%%%%%%%

\begin{cbunit}

 \onecolumngrid
 \begin{center}
 \textbf{\large Supplementary Material for \\ Permutation Symmetry Determines  the Discrete Wigner Function}
 \end{center}
 \twocolumngrid
 
 In this supplementary material, we prove Lemma~3 in the main text and show that the minimal degree of nontrivial irreducible
projective representations of any transitive subgroup $H$ of $\Sp{2n}{2}$
 with $n\geq 2$ is always larger than $2^n/2$ except when $n=2$  and $H$
is isomorphic to $\Sp{2}{2^2}$. Both results are needed in the proof of Theorem~3. To this end, we need to list  transitive linear groups on vector spaces over finite fields  as determined by Hering
\cite{Heri85} and Liebeck \cite{Lieb87}; see also \rcite{DixoM96book} and
Table 7.3 in \rcite{Came99book} (note that case~8 in the following lemma is missing in the table but appears in \rcite{Lieb87}).

\begin{lemma}[CFSG]\label{lem:TransitiveGen}
Any transitive subgroup $H$ of $\GL{2n}{p}$ with $n\geq 2$ satisfies one of
the conditions below:
\begin{enumerate}\itemsep-0.5ex
\item $\Sp{2m}{p^k} \leq H$ with  $1\leq m\leq n$ and $mk=n$.

\item $\SL{m}{p^k}\leq H\leq \GaL {m}{p^k}$ with $m\geq3$ and $mk=2n$, where
$\GaL {m}{p^k}$ is the general semilinear group.

\item $H\leq \GaL {1}{p^{2n}}$.

\item $p=2$, $n=3k$, and $G_2(p^k)\trianglelefteq H$, where  $G_2(p^k)$ is
an exceptional group of Lie type \cite{Wils09book}.

\item $p=2$, $n=2$,  and $H\simeq A_6$ or $H\simeq A_7$, where $A_6$ and
$A_7$ are alternating groups of degrees 6 and 7.

\item $p=2$, $n=3$, and $H\simeq \PSU{3}{3}$.

\item $p=3$, $n=2$, and  $2^{1+4}\trianglelefteq H$, where $2^{1+4}$ denotes
an extraspecial 2-group \cite{KurzS04book} of this order.

\item $p=3$, $n=2$, and  $\SL{2}{5}\trianglelefteq H$.

\item $p=3$, $n=3$, and $H\simeq \SL{2}{13}$.
\end{enumerate}
\end{lemma}
\begin{remark}
Here  $\Sp{2m}{p^k}$, $\SL{m}{p^k}$, $\GaL {m}{p^k}$, and $G_2(p^k)$ are
extension-field-type subgroups of  $\GL{2n}{p}$. Complete classification
is also available in the case  $n=1$ but is not necessary here.
Quite surprisingly, barring a few exceptions in the cases $p=2,3$ and $n=2,3$, there are only three families of transitive subgroups for odd prime~$p$ and four families for $p=2$. Moreover, for subgroups of
$\Sp{2n}{p}$, cases 2 and 3 cannot happen  according to  the proof of Lemma~3 below. So there is only one (two) family of transitive subgroups for odd (even) prime~$p$.  
\end{remark}

% \begin{lemma}[CFSG]\label{lem:TransitiveOdd}
% Any transitive  subgroup $H$ of $\Sp{2n}{p}$ for odd prime $p$ contains %the
% central involution.
%\end{lemma}

Before proving Lemma~3, we need to introduce several additional concepts.
A \emph{Singer cyclic group} of a classical group over a finite field is
an irreducible cyclic subgroup of maximal order \cite{Hupp70, Bere00}. Any
generator of a Singer cyclic group is called a \emph{Singer cycle}. Singer
cyclic groups of $\GL{n}{p}$,  $\SL{n}{p}$, and $\Sp{2n}{p}$ have orders
$p^n-1$, $(p^n-1)/(p-1)$, and $p^n+1$, respectively. In all these cases,
all Singer cyclic groups of a given group are conjugate to each other. We
are only concerned with Singer cyclic groups (cycles) of $\Sp{2n}{p}$ except
when stated otherwise.
Let $a,b$ be positive integers with $b>1$. A prime~$r$ dividing  $b^a-1$
is a \emph{Zsigmondy prime} (also known as  primitive prime divisor) \cite{Zsig1892,Roit97}
if $r$ does not divide $b^j-1$ for $j=1, 2,\ldots, a-1$.
It is known that $b^a-1$ has a Zsigmondy prime except when $(b,a)=(2,6)$,
or $a=2$ and $b+1$ is a power of 2.
An element  of $\Sp{2n}{p}$ is a \emph{Zsigmondy cycle} if its order is 
a Zsigmondy prime of $p^{2n}-1$. The group generated by a Zsigmondy cycle
is called a Zsigmondy cyclic group. All Zsigmondy cyclic groups of a given
order are conjugate to each other.
The centralizer of a Zsigmondy cycle is a Singer cyclic group. When $p$ is
odd,  the centralizer  contains a unique involution, that is, the central
involution  of $\Sp{2n}{p}$ \cite{Hupp70, Shor92book, Bere00, Zhu15Sh}.

Now we are ready to prove Lemma~3 in the main text that any transitive  subgroup $H$ of $\Sp{2n}{p}$ for odd prime~$p$ contains the
central involution.
\begin{proof}[Proof of Lemma~3]

When $n=1$,  the group   $\Sp{2n}{p}\simeq\SL{2}{p}$ has a unique involution
\cite{Zhu10}. Since $H$ has even order, it must contain the  involution.

When  $n\geq2$, the group $H$ must satisfy one of the conditions 1, 2, 3,
7, 8, 9 in \lref{lem:TransitiveGen}.
We shall show that $H$ contains the central involution in cases 1, 7, 
8, and 9, while cases 2 and 3 cannot happen.

In case 1,  the center of $\Sp{2m}{p^k}$ coincides with that of $\Sp{2n}{p}$,
so $H$ contains the central involution.

Case 2 cannot happen because   $\Sp{2n}{p}$ cannot contain $\SL{m}{p^k}$
with $3\leq m\leq 2n$ and $mk=2n$. Note that each Singer cycle of $\SL{m}{p^k}$
has  order $(p^{2n}-1)/(p^k-1)$, but $\Sp{2n}{p}$ has no such element~\cite{Hupp70,
Bere00}.

Case 3 can be ruled as follows. The group $\GaL{1}{p^{2n}}$ is a semidirect
product
$\mrm{Gal}(\mathbb{F}_{p^{2n}}/\bbF_{p})\ltimes \GL{1}{p^{2n}}$, where $\GL{1}{p^{2n}}$
can be identified as   a Singer cyclic group of $\GL{2n}{p}$, and $\mrm{Gal}(\mathbb{F}_{p^{2n}}/\bbF_{p})$
is the Galois group of the field $\mathbb{F}_{p^{2n}}$ over the prime field
$\bbF_{p}$, which  is cyclic of order $2n$ \cite{DixoM96book}. The intersection
of any Singer cyclic group of $\GL{2n}{p}$ with $\Sp{2n}{p}$ has order at
most $p^n+1$ \cite{Bere00}. So the order of $H$ is at most $2n(p^n+1)$ and 
cannot be divisible by $p^{2n}-1$. Consequently, $H$ cannot be transitive.

If case 7 holds, let $P$ be the normal subgroup of $H$ that is isomorphic
to $2^{1+4}$. Then the center of $P$ has order 2 and is contained in the
center of $H$. On the other hand, $H$ has order divisible by a Zsigmondy
prime of $3^4-1$ so it contains a Zsigmondy cycle. The involution in the
center of $P$ commutes with the Zsigmondy cycle, so it must coincide with
 the central involution of $\Sp{2n}{p}$.

In case 8, the conclusion follows from the same reasoning as  in case 7 since
the center of $\SL{2}{5}$  has order~2. 

If case 9 holds, then $H$ contains a Zsigmondy cycle; the unique involution
in $H$ commutes with the Zsigmondy cycle and thus coincides with the central involution
of $\Sp{2n}{p}$.
\end{proof}

Next, we show that the minimal degree of nontrivial irreducible projective representations of any transitive subgroup $H$ of $\Sp{2n}{2}$  with $n\geq 2$ is always larger than $2^n/2$ except when $n=2$  and $H$ is isomorphic to $\Sp{2}{2^2}$. Recall that $H$ satisfies one of the conditions 1 to 6 in \lref{lem:TransitiveGen}.
  Cases 2 and 3 cannot happen according to the same reasoning
as in the proof of Lemma~3.

In case 1, according to Table~II in \rcite{TiepZ96},  the minimal degree of  nontrivial irreducible projective representations of $\Sp{2m}{2^k}$ with $mk=n$ and $n\geq3$  is
\begin{equation}\label{eq:MinDeg1}
\begin{cases}
2^n-1, &  k=n;\\
\frac{(2^n-1)(2^n-2^k)}{2(2^k+1)}, & k<n;
\end{cases}
\end{equation}
which is always larger than $2^n/2$. It remains to consider the cases  $n=2$ and $k=1, 2$. The two groups $\Sp{4}{2}$ and $\Sp{2}{2^2}\simeq\SL{2}{2^2}$ are isomorphic to the symmetric group of degree~6 and alternating group of degree~5, respectively. The minimal degree of  nontrivial irreducible projective representations is 4 for $\Sp{4}{2}$ and 2 for $\Sp{2}{2^2}$.
Incidentally, according to theorem 7 in \rcite{BoltRW61II}, the HW group in dimension $2^n$ with $n\geq 2$ is not complemented in the  Clifford group, so  the Clifford group  has no subgroup isomorphic to $\Sp{2n}{2}$.

In case 4, according to   Sec. 5.3 in  \rcite{Lube01},   the minimal degree of nontrivial irreducible projective representations of  $G_2(2^k)$ (assuming $n=3k$) is
\begin{equation}\label{eq:MinDegG2}
\begin{cases}
6, &  k=1; \\
2^n-1, & 2\nmid k, k\geq 3;  \\
2^n+1, & 2 | k, k\geq 2;
\end{cases}
\end{equation}
which is always larger than $2^n/2$.

In case 5,  $\Sp{4}{2}$ cannot contain $A_7$ (which has larger order than $\Sp{4}{2}$). The minimal degree of nontrivial irreducible projective representations of $A_6$ is 3. Incidentally,  any subgroup of the Clifford group in dimension 4 that is isomorphic to $A_6$ is irreducible.

In case 6, the two minimal degrees of nontrivial irreducible projective representations of $\PSU{3}{3}$ are  6 and 7 according to Table~V in \rcite{TiepZ96}. Incidentally, the stabilizer of each fiducial state of the set of Hoggar lines is isomorphic to $\PSU{3}{3}$ \cite{Zhu15S}.

\end{cbunit}
\end{document}